\documentclass[10pt, conference,twocolumn]{ieeeconf}

\IEEEoverridecommandlockouts        

\usepackage[english]{babel} 

\let\labelindent\relax
\usepackage{enumitem} 

\usepackage{graphics} 
\usepackage{epsfig} 
\usepackage{amsmath} 

\IEEEoverridecommandlockouts        

\usepackage[english]{babel} 

\usepackage{enumitem} 

\usepackage{graphics} 
\usepackage{epsfig} 
\usepackage{amsmath} 
\usepackage{amssymb}  
\usepackage{dsfont} 
\usepackage{amsthm}

\newcommand{\mbf}[1]{\mathbf{#1}}
\newcommand{\bb}[1]{\mathbb{#1}}
\newcommand{\ml}[1]{\mathcal{#1}}
\newcommand{\mb}[1]{\boldmath{#1}}
\newcommand{\mc}[1]{\mathcal{#1}}

 \usepackage{mathtools}

\DeclareMathOperator*{\argmin}{arg min}    

\newtheorem{theorem}{Theorem}
\newtheorem{remark}{Remark}
\newtheorem{prop}{Proposition}

\graphicspath{{graphics/}}

\usepackage[dvipsnames]{xcolor}
\usepackage[deletedmarkup=sout,authormarkup=superscript]{changes}

\definechangesauthor[name={Sophie}, color=NavyBlue]{SH}
\definechangesauthor[name={Dominic}, color=RedViolet]{DLM}
\definechangesauthor[name={Bepe}, color=OliveGreen]{GB}

\usepackage{cite}
\title{\LARGE \bf
Receding Horizon Games with Coupling Constraints\\ for Demand-Side Management}

\author{Sophie Hall, Giuseppe Belgioioso, Dominic Liao-McPherson, and Florian D\"{o}rfler
\thanks{The authors are with the ETH Zürich Automatic Control Lab, 8092 Zürich, Switzerland. {Emails: \texttt{\{shall, gbelgioioso, dliaomc, dorfler\}@ethz.ch}}. This work was supported by the Swiss National Science Foundation through NCCR Automation (Grant Number 180545).}
        }

\begin{document}

\maketitle
\thispagestyle{empty}
\pagestyle{empty}

\begin{abstract}
Distributed energy storage and flexible loads are essential tools for ensuring stable and robust operation of the power grid in spite of the challenges arising from the integration of volatile renewable energy generation and increasing peak loads due to widespread electrification.
This paper proposes a novel demand-side management policy to coordinate self-interested energy prosumers based on \textit{receding horizon games}, i.e., a closed-loop receding-horizon implementation of game-theoretic day-ahead planning.
Practical stability and recursive constraint satisfaction of the proposed feedback control policy is proven under symmetric pricing assumptions using tools from game theory and economic model predictive control. Our numerical studies show that the proposed approach is superior to standard open-loop day-head implementations in terms of peak-shaving, disturbance rejection, and control performance.
\end{abstract}

\section{Introduction}
Global electricity demand is predicted to increase by nearly 30\% from 2020 to 2030, with the share of solar and wind energy growing from under 10\% to 23\% \cite{iea2021world}. This increasing demand along with the volatility of renewables is leading to larger peak loads in many distribution grids. However, improving the physical infrastructure to handle this unprecedented increase in both energy demand and peak loads is extremely expensive. On the other hand, the widespread deployment of sensing, communication, and actuation technologies, such as smart meters, and the proliferation of local storage and generation offers a cheaper alternative for reducing peak loads while increasing the resilience of the power grid \cite{parag2016electricity}.

%
%

In the US and Europe there has been a rapid proliferation of \textit{prosumers}: consumers who produce and store energy locally in addition to drawing power from the main grid. In the future, these distributed storage and generation devices will allow prosumers to reduce the amount of energy they draw from the main grid, shift what remains to off-peak hours, or even  provide energy to other prosumers in a process known as demand-side management (DSM). In privatized grids, e.g., the Texas interconnect, prosumer behaviour is not directly controllable, and they must be incentivized to participate in contributing to grid stability. 

Game theory has emerged as a promising framework for designing mechanisms that incentivize self-interested prosumers to participate in ensuring safe grid operation while pursuing their local economic objectives. There is an extensive literature on game-theoretic DSM schemes for self-interested prosumers\cite{atzeni2013demand, li2017efficient, jo2020demand,belgioioso2021operationally, cadre2020peer}. A typical approach is to incentivize load shifting by dynamically changing the electricity price and enforcing safe grid operation through operational limits on both lines and aggregate loads \cite{atzeni2013demand, li2017efficient, jo2020demand}. The grid operation problem is then formulated as a game played between the prosumers who attempt to meet their energy needs as (cost) efficiently as possible while respecting grid constraints. The load profile is then computed by finding a suitable game-theoretic equilibrium between the prosumers, e.g., a generalized Nash equilibrium \cite{belgioioso2021operationally, cadre2020peer}. The equilibrium is called generalized as prosumers' decisions are coupled through shared constraints, i.e., limits on the aggregate energy demand.


Most existing game-theoretic DSM schemes, e.g., \cite{atzeni2013demand,belgioioso2021operationally}, perform day-ahead planning wherein prosumers plan overnight for the upcoming day based on demand and generation forecasts and commit to executing that plan (with deviation often resulting in a financial penalties). In control-theoretic terms, this corresponds to repeated open-loop control over a 24-hour horizon. These schemes are motivated by the existing day-ahead energy markets but in practice, such schemes are inefficient as, on a given day, prosumers have no information about tomorrow's consumption and prices. This leads to undesired ``end-of-day" effects, namely, prosumers tend to discharge their batteries and reduce their load on the grid at the end of the planned horizon \cite{jo2020demand, cadre2020peer, belgioioso2021operationally}. Furthermore, open-loop control strategies cannot react to unexpected disturbances, such as inaccurate forecasts of renewable generation, sudden spikes in the passive load (e.g., due to heat waves, etc.), or decreases in the power available from the main grid (e.g., line faults).

In fully-cooperative settings, such challenges are typically tackled using receding-horizon control schemes such as (multi-agent) model predictive control (MPC), which offers a powerful paradigm for optimal control of constrained systems. There exists an extensive literature on MPC schemes for DSM \cite{zong2012application, silvente2015rolling, mahdavi2017model}. A receding-horizon implementation is suitable for future energy markets in which local, decentralized real-time trading is predicted to play an important role \cite{parag2016electricity}. However, such MPC approaches are fully-cooperative, namely, they assume prosumers are working towards a common goal (the social welfare), and cannot capture the self-interested nature of prosumers. 

 
To overcome the aforementioned limitations of repeated open-loop DSM and fully-cooperative MPC schemes, in this paper we propose an MPC-inspired game-theoretic DSM scheme which we refer to as a Receding Horizon Game (RHG). Our contributions are threefold:
\begin{enumerate}
\item[(i)] We propose a game-theoretic MPC mechanism for DSM in which at each time step: (1) a generalized game over a prediction horizon is solved to obtain the optimal storage and consumption profile of each prosumer that are also jointly operationally-feasible for the distribution grid, i.e., the aggregate-load limits are respected; (2) each prosumer applies the first control input of the planned profile; (3) finally, the prediction horizon is forward-shifted and the procedure repeats;

\item[(ii)] We prove closed-loop stability of the proposed policy under the reasonable assumption that utility electricity prices are uniform across the population of prosumers by combining potential games with economic MPC; and
\item[(iii)] We show via numerical simulations with real data the superior performances of RHG over day-ahead optimization for peak shaving and successful disturbance rejection in a contingency scenario where the aggregate load supplied by the grid drops by $-60\%$.
\end{enumerate}

Our approach is related to others in the literature. The authors in \cite{stephens2015game} also consider a game-theoretic MPC approach for DSM, but do not consider any system-wide coupling constraints. Enforcing these system-wide constraints in the presence of disturbances is essential for safe grid operation. Moreover, they do not provide any closed-loop stability or constraint satisfaction results. A receding-horizon generalized game approach is also adopted in \cite{scarabaggio2022distributed} to solve the DSM problem with uncertainty in wind power forecasting, however no convergence analysis or stability guarantees are given. To account for inaccurate forecasts, \cite{estrella2019shrinking} proposes a ``shrinking-horizon'' DSM scheme which however still suffers from ``end-of day'' effects. The authors of \cite{paola2018distributed, fele2018framework} propose a receding-horizon framework for electric load scheduling. Their solution differs substantially from ours as they consider periodic Wardrop equilibria of an aggregative game without any coupling constraints while we consider Nash equilibria with system-wide coupling constraints. Game-theoretic MPC approaches have also been proposed for other applications such as autonomous driving/racing \cite{spica2020real,cleach2022algames} and highway traffic control \cite{cenedese2021highway}, without any stability certificates. Finally, RHGs are an extension of multi-agent economic MPC \cite{mueller2017economic} that relaxes the assumption that agents are fully cooperative. 


\textit{Notation}: We denote by $\bb{Z}_N$ a sequence of $N$ non-negative integers $\bb{Z}_N = \{0,\dots, N\!-\!1\}$. Given a set $\mc A:=\{1, \ldots,M \}$ of $M$ agents labelled by $v\in \mathcal A$. We denote the stacked vector of all agents' decisions by  $u = \text{col}(\{u^v\}_{v\in \ml{A}}):= [(u^1)^\top, \ldots, (u^M)^\top]^\top$, where $u^v$ is the decision vector of agent $v$, and by $u^{-v}$ the decision of all agents except agent $v$, i.e., $u^{-v} = \text{col}(\{u^{s}\}_{s\in \ml{A}\backslash v})$. Given $N$ matrices, $H_1 , . . . , H_N$, $\text{blkdiag}(H_1 , . . . , H_N )$ denotes the block diagonal matrix with $H_1 , . . . , H_N$ on the main diagonal. The zero column vector of dimension $N$ is denoted as $\mbf{0}_{N}$. Our use of class $\mc{K}$, $\mc{KL}$, and $\mc{K}_\infty$ comparison functions follows \cite[\S 1.2]{faulwasser2018economic}. 

%

%

\section{Modelling}

We consider a distribution grid composed of $M$ active prosumers $v \in \ml{A} = \{1,\dots, M\}$ connected to the main transmission grid via a point of common coupling. 
Each active prosumer $v\in \mc{A}$ consumes $e^v_t$ and stores $s^v_t$ units of power, during each time instant $t$. A subset of prosumers $\ml{A}_g\subset \ml{A}$ can additionally generate $g^v_t$ units of power using non-dispatchable generation units, such as solar or wind based generators.  
For these non-dispatchable units, the generation at each time $t$ is solely determined by external factors, such as the weather. 

To fulfil their energy needs, prosumers can buy energy from the main grid. Their load on the grid at time $t$ is denoted as $l_t^v \in \bb{R}$ and given by

\begin{align} 
l_t^v = e_t^v + s_t^v -g_t^v, \quad \forall v \in \ml{A} \label{eqn:LoadGen},
\end{align}
with $g_t^v = 0, \;  \forall v \notin \ml{A}_g$. Furthermore, we define a set of passive consumers $v \in \ml{P}$ that do not participate in the DSM program but still contribute to the aggregate load on the grid

\begin{align} \label{eqn:AggrLoad}
L_t = L_t^{\ml{A}}  + L_t^{\ml{P}} , \quad
L_t^{\ml{A}} = \sum_{v\in \ml{A}}l_t^v, \quad 
L_t^{\ml{P}}  = \sum_{v\in \ml{P}} l_t^v, 
\end{align}
with $L_t^{\ml{A}} \in \bb{R}$ and $L_t^{\ml{P}} \in \bb{R}$ denoting the aggregate load from active prosumers and passive consumers, respectively. 

\subsection{Energy storage}
\noindent The battery of each prosumer $v\in \ml{A}$ follows the dynamics

\begin{subequations}
\begin{align} \label{eqn:BatteryDyn}
q_{t+1}^v = \alpha^v q_t^v + \beta^v s_t^v,  \quad \forall v \in \ml{A},
\end{align}
where $q_t^v$ is the state-of-charge (SoC) and $s_t^v$ is a controllable input which indicates charging for $s_t^v > 0$ and discharging for $s_t^v < 0$. The parameters $\alpha^v \in \left[0, 1\right]$ and $\beta^v \in \left[0, 1 \right]$ are the leakage rate and the charging efficiency, respectively. Each battery is subject to the following constraints on their storage capacity and charging rate:

\begin{align}  \label{eqn:CapacityCon}
0 \leq q_t^v \leq \bar{q}^v, \quad \forall v \in \ml{A}\\ \label{eqn:ChargeCon}
\underline{s}^v\leq \beta^v s_t^v \leq \bar{s}^v, \quad \forall v \in \ml{A}
\end{align}
\end{subequations}
where $\bar{s}^v$ and $\underline{s}^v$ are the upper and lower charging limits, and $\bar{q}^v$ is the storage capacity.

\subsection{Flexible energy consumption}
\noindent Prosumers are willing to shift their load $l_t^v$ not only by using their local storage but also by adapting their energy consumption $e_t^v$. Prosumers have an inflexible baseline consumption, e.g., the energy needed for domestic appliances, and a flexible consumption, e.g., electric vehicle charging. The minimum and maximum consumption at every hour are modelled via the following constraints:

\begin{align} \label{eqn:ConsLim}
\underline{e}^v \leq e_t^v \leq \bar{e}^v, \quad \forall v \in \ml{A} ,
\end{align}
where $\bar{e}^v > \underline{e}^v \geq 0$ are the consumption bounds.  

Typically, prosumers are willing to shift their flexible consumption to off-peak hours but not to reduce their total daily power consumption \cite{li2017efficient, jo2020demand}. To model these limits on consumption flexibility, we introduce an energy shift state $\zeta_t^v$ which integrates the deviation from the nominal consumption $e_t^{v,\text{ref}}$, i.e., the amount of energy that would be consumed without DSM. The dynamics of $\zeta_t^v$ are given by

\begin{equation} \label{eqn:Integrator}
\zeta_{t+1}^v =  \zeta_t^v + (e_t^v-e_t^{v,\text{ref}} \,), \quad \forall v \in \ml{A}.
\end{equation}
The shift state $\zeta_t^v$ can be interpreted as a consumption debt, if $\zeta_t^v <0$, or credit, if $\zeta_t^v >0$. In practice, prosumers are willing to shift their consumption by a limited amount, thus motivating the following box constraints on the shift state: 

\begin{equation} \label{eqn:DebtLim}
\underline{\zeta}_t \leq \zeta^v_t \leq \bar{\zeta}_t, \quad \forall v \in \ml{A}.
\end{equation}
Picking $\underline{\zeta}_t = \bar{\zeta}_t = 0$ once in a period (i.e., in the next 12 hours or the next day) ensures that the total energy consumption over that period remains constant. 


\subsection{Load on utility grid}
\noindent The power supplied to individual prosumers is limited, e.g., by fuses in their homes, total storage capacity or load, thus motivating the following constraints:

\begin{equation} \label{eqn:LoadLim}
0 \leq l_t^v \leq \bar{l}^v, \quad \forall v \in \ml{A},
\end{equation}
where $\bar{l}^v$ is the maximum power a prosumer can absorb from the main grid at every time-instant. 

Power line and transformer constraints at the point of common coupling limit the total power that the main grid can supply to the whole distribution network. We model this using the following constraint on the aggregate load $L_t$:

\begin{equation} \label{eqn:AggrLoadLim}
\underline{L}_t \leq L_t \leq \bar{L}_t,
\end{equation}
where $\bar{L}_t>\underline{L}_t>0$.

\subsection{Self-interested prosumer model} \noindent
Each prosumer $v \in \ml{A}$ is self-interested and aims to minimize its electricity bill, the operational cost of its battery, and the discomfort from shifting its energy consumption, subject to the operational limits of its devices but is also incentivized to help enforce system-wide coupling constraints (as stability of the grid is in the best interest of all prosumers) \cite{atzeni2013demand, belgioioso2021operationally}.

Denote each prosumer $v$'s stacked control vector by $u^v_t = (e_t^v, s_t^v)$, which collects the energy consumption and the battery charging/discharging inputs, and state vector by $x_t^v = (\zeta_t^v, q_t^v)$, which collects the battery SoC and the energy debt. Then, the dynamics of each prosumer $v$ can be cast as an (linear time-invariant) LTI system of the form

\begin{align} \label{eqn:LTIDynamics}
x^v_{t+1} = A^v x^v_t + B^v u^v_t + d_t^v,
\end{align}
with the following system matrices and disturbance vector

\begin{align*}
A^v = \begin{bmatrix}
1& 0\\
0 & \alpha^v
\end{bmatrix}, \quad
B^v = \begin{bmatrix}
1& 0\\
0 & \beta^v
\end{bmatrix}, \quad
d_t^v = \begin{bmatrix}
-e_t^{v,\text{ref}}\, \\
0
\end{bmatrix}.
\end{align*}
The local cost function for each prosumer is of the form

\begin{align} \label{eqn:StageCost}
\ell_t^v(x^v_t,u^v_t,u_t^{-v}) =\,  \underbrace{\sigma^v(L_t) \,l_t^v}_{\text{energy cost}} + \underbrace{\gamma_{1,t}^v \, \zeta_t^2}_{\text{energy shift}} +  \underbrace{\gamma_{2,t}^v \, q_t^2}_{\text{battery usage}},
\end{align}
where $\sigma^v(L_t)$ is the price of electricity, and $\gamma_{1,t}$ and $\gamma_{2,t}$ are positive weights. Similarly to \cite{atzeni2014noncooperative}, we model the energy price $\sigma^v$ as an affine function of the total demand on the grid $L_t$, i.e.,

\begin{align} \label{eqn:EnergyPrice}
\sigma^v(L_t) = \rho_{1,t}^v\; L_t + \rho_{2,t}^v,
\end{align}
where $\rho_{1,t}^v, \rho_{2,t}^v$ are positive constants representing different price rates that prosumers previously negotiated with suppliers. The second and third terms in (\ref{eqn:StageCost}) account for the discomfort each prosumer experiences when shifting its consumption and the desire to minimize the usage of its storage unit to avoid degradation, respectively.

\section{Game-theoretic MPC}

\subsection{Problem formulation}
\begin{figure*}
        \centering
        \includegraphics[width=0.99\textwidth]{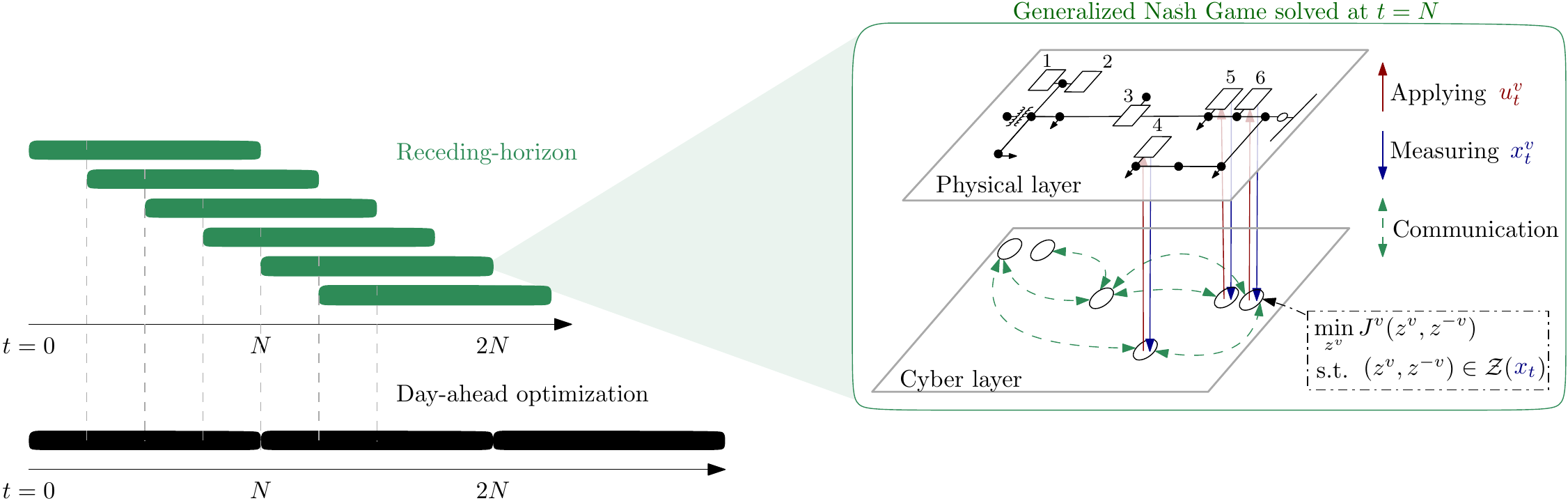} 
        \caption{On the left, the receding-horizon strategy (top) is contrasted with the day-ahead optimization (bottom). On the right, the generalized game solved at $t=N$ is shown, displaying the interaction between the physical system of every prosumer and the cyber layer with the v-GNE computation.}
\label{fig:RHGschematic}
\end{figure*} 

Typically, DSM schemes require the prosumers to first solve a planning problem and then commit to the resulting ``optimal" action profiles over a certain period of time into the future, i.e., over a prediction horizon $\bb{Z}_N = \{0, \dots,N-1\}$. A common choice is day-ahead optimization, i.e., $N = 24$, with a sampling period of one hour\cite{atzeni2014noncooperative, cadre2020peer}. We  assume that an accurate forecast of the non-dispatchable generation $g_t^v$ is available over the prediction horizon.

In day-ahead optimization, each prosumer $v \in \ml{A}$ aims to minimize their cumulative cost over the next 24 hours by solving the following optimal control problem (OCP):

\begin{subequations}
\label{eqn:MPCagent}
\begin{align} \label{eqn:RunningCost}
&&\displaystyle \min_{u^v,\, x^v} &\; \displaystyle \sum_{k\in \bb{Z}_N} \ell_k^v(x^v_k,u^v_k,u_k^{-v})\\ \label{eqn:Constr1}
&&\textrm{s.t.} \quad &  x_{k+1}^v =  A^v x^v_k + B^v u^v_k + d_k^v,\;   k \in \bb{Z}_{N}\\
\label{eqn:Constr2}
&&& u_k^v \in \ml{U}_k^v  \cap \ml{C}_{k}(u_k^{-v}), 
\hspace*{3.7em}
k \in \bb{Z}_N 
\\ 
&&& x_k^v \in \ml{X}_k^v, \ x_0^v = \mbf{x}^v,
\hspace*{3.9em}
 k \in \bb{Z}_{N+1}
\label{eqn:Constr3}
\end{align}
\end{subequations}
where $\mbf{x}^v$ is the initial state of prosumer $v$ and 

\begin{align} \label{eqn:ConstrSet1}
\ml{X}_{k}^v &= \{ x_k^v\, | \,  (\ref{eqn:CapacityCon}), (\ref{eqn:DebtLim}) \; \text{hold} \},\\ \label{eqn:ConstrSet2}
\ml{U}_{k}^v&= \{u_k^v\,|\, (\ref{eqn:LoadGen}), (\ref{eqn:ChargeCon}), (\ref{eqn:ConsLim}),(\ref{eqn:LoadLim}) \; \text{hold} \},\\ \label{eqn:ConstrSet3}
\ml{C}_k(u_k^{-v}) &= \{u_k^v \, | \, (\ref{eqn:AggrLoadLim}) \; \text{holds}  \}.
\end{align} 
The $M$ OCPs in (\ref{eqn:MPCagent}) are coupled in the stage cost \eqref{eqn:RunningCost} as the energy price $\sigma^v(L_t)$ depends on the aggregate load of all $M$ prosumers and in the input constraints \eqref{eqn:Constr3} due to the aggregate load limits. Together, these inter-dependent OCPs form a generalized game\footnote{The information structure of the game is discussed in Remark \ref{remark:information}.}, i.e., an equilibrium problem where the cost and feasible set of each prosumer depend on other prosumer’s decisions \cite{facchinei2009generalized}.


To recast \eqref{eqn:MPCagent} in a more compact form, we define the stacked vector $z^v = (u^v, x^v )$ and a vector $\phi$ that collects all of the exogenous parameters in (\ref{eqn:MPCagent}), i.e., $\phi \!=\! \text{col}(\phi_1, \phi_2) $ with $\phi_1 = \text{col}(\{\bar{\zeta}_k^v, \underline{\zeta}_k^v\}_{k\in\bb{Z}_{N+1}, v\in \ml{A}})$ and $\phi_2 = \text{col}(\{d_k^v, g_k^v , \bar{L}_k,  \underline{L}_k,  L_k^{\ml{P}}, \rho_{k,1}^v, \rho_{k,1}^v\}_{k\in\bb{Z}_N, v\in \ml{A}})$.

Using this notation, we combine the constraint sets of all prosumers (\ref{eqn:ConstrSet1})-(\ref{eqn:ConstrSet3}) in a global action set 

\begin{align}
 \ml{Z}(\mbf{x}, \phi) =  \{  \forall v\in \ml{A},\,(u^v, x^v)\, | \, \eqref{eqn:Constr1} -\eqref{eqn:Constr3} \; \,\text{hold} \},
\end{align}
and compactly rewrite the coupled OCPs \eqref{eqn:MPCagent} in the following standard form for generalized games:
 
\begin{equation} \label{eqn:GamePerAgent}
\forall v\in \ml{A}:\quad \left\{
\begin{array}{rl}
 \displaystyle \min_{z^v}& {J^v(z^v, z^{-v},\phi)}\\
\textrm{s.t.} & (z^v, z^{-v}) \in \ml{Z}(\mbf{x}, \phi),
\end{array}
\right. 
\end{equation}
where $J^v(z^v, z^{-v},\phi)$ corresponds to (\ref{eqn:RunningCost}). 
\subsection{Solution concept}

A meaningful solution concept for \eqref{eqn:GamePerAgent} is the generalized Nash equilibrium (GNE), i.e., a set of strategies $\bar{z} = \text{col}(\bar{z}^{v})_{v \in \ml{A}}$ for which no prosumer $v \in \ml{A}$ can reduce its cost by unilaterally changing its strategy \cite[\S 2]{facchinei2009nash}, i.e., $\forall v \in \ml{A}$:

\begin{align*}
J^v( \bar{z}^{v}, \bar{z}^{-v}, \phi) \leq J^v(z^v, \bar{z}^{-v},\phi), \ \forall\, (z^v,\bar{z}^{-v}) \in  \ml{Z}(\mbf{x},\phi).
\end{align*}
Here, we target the subclass of variational GNEs (v-GNEs) which correspond to the solutions of the following parametrized generalized equation \cite[Prop. 12.4]{facchinei2009nash}:

\begin{align} \label{eqn:VI}
\text{F}(\bar z,\phi) + \ml{N}_{\ml{Z}}(\bar z,\mbf{x},\phi) \ni 0,
\end{align}
where F$(z, \phi) = \text{col}(\nabla_{z^v} J^v(z^v, z^{-v},\phi))_{v\in \ml{A}}$ is the pseudo-gradient of \eqref{eqn:GamePerAgent} and $\ml{N}_{\ml{Z}}$ is the normal cone \cite[Def.~6.38]{bauschke2017convex} of the global action set $ \ml{Z}(\mbf{x}, \phi)$. 
The mapping from initial state $\mbf{x}$ and parameters $\phi$ to the solution of the generalized equation (\ref{eqn:VI}) is

\begin{align} \label{eqn:SolMap}
\ml{S}(\mbf{x}, \phi) = \{z ~|~  \text{F}(z, \phi) + \ml{N}_{\ml{Z}}(z, \mbf{x},\phi) \ni 0 \}.
\end{align} 

\noindent Variational GNEs are useful for grid operation as they satisfy the operational constraints, are strategically (Nash) stable, i.e., no prosumer has an incentive to deviate from their agreed upon input profile, and are ``economically fair" equilibria, in the sense that each prosumer incurs the same marginal loss due to the presence of the coupling constraints \cite{facchinei2009nash}.


\begin{remark} \label{remark:information}
There exist several algorithms to compute a v-GNE of \eqref{eqn:MPCagent}, i.e., an element of \eqref{eqn:SolMap}, based on different communication structures such as semi-decentralized \cite{paccagnan2018nash, belgioioso2021semi} or fully-distributed \cite{yi2019operator, bianchi2022fast}.{\hfill$\square$}
\end{remark}

\subsection{Implementation}

Most game-theoretic DSM schemes in the literature are employed in a day–ahead manner \cite{atzeni2014noncooperative}. That is, once a day (typically at midnight) the prosumers compute their storage and consumption profiles $\bar{z}\in \ml{S}(\mbf{x}, \phi)$ for the upcoming day by finding a v-GNE for the game (\ref{eqn:MPCagent}) with $\mbf{x}$ being the global system state at the time of computation and $\phi$ the vector of exogenous parameters which includes the consumption and generation forecasts for the upcoming day. Then, each prosumer applies the resulting trajectory ${\bar{z}^{v}}$ in an open-loop manner over the next 24 hours, before the whole process is repeated the next day. However, such an open-loop approach leads to undesirable ``end-of-day" effects in which prosumers significantly change their strategy towards the end of the finite-horizon which leads to unrealistic outcomes, e.g., fully discharging the batteries at the end of every day. Furthermore, open-loop approaches cannot respond to sudden disturbances, such as line faults or sudden spikes in the passive load due to e.g., heat waves. Such events may lead to inefficient operation or, worse, to blackouts or infrastructure damage caused by constraint violations. 

To overcome these drawbacks, we employ a receding-horizon implementation inspired by MPC. At time $t$, the prosumers compute a v-GNE $\bar{z} \in \ml{S}(x_t, \phi_t)$ of (\ref{eqn:MPCagent}), then, each prosumer $v$ applies the first element $\bar{z}^{v}_0$ of the predicted control trajectory. This creates a feedback policy\footnote{Under certain assumptions given in Proposition \ref{prop:Potential}, the solution of \eqref{eqn:VI} is unique and therefore \eqref{eqn:SolMap} is a singleton.}
  
\begin{align}\label{eqn:FeedbackLaw}
u^{v}_t =  \kappa^v(\mbf{x}_t, \phi_t)   = \Xi^v \ml{S}(\mbf{x}_t, \phi_t) ,
\end{align}
where $\Xi^v$ is a selection matrix that extracts the first input of the control sequence of prosumer $v$, $u_0^v$, from $\ml{S}(\mbf{x}_t, \phi_t)$ defined in \eqref{eqn:SolMap}. The resulting closed-loop system is

\begin{align}\label{eqn:System}
x_{t+1}^v = A^v x_t^v + B^v \kappa^v(\mbf{x}_t, \phi_t) + d^v_t, \quad \forall v \in \ml{A}.
\end{align}

The difference between the day-ahead optimization and RHG policies are illustrated in Figure \ref{fig:RHGschematic}.

\section{Closed-loop Stability}

In this section, we show that under symmetric pricing conditions the closed-loop system is recursively feasible and admits an asymptotically stable equilibrium point for every constant set of parameters $\phi$.
We begin by showing that the DSM game in (\ref{eqn:MPCagent}) is a generalized potential game \cite{facchinei2010decomposition} whenever the energy price is the same for each prosumer. In a potential game the equilibria coincide with the minimizers of a global optimization function \cite[Def. 2.1]{facchinei2010decomposition}. 

\begin{prop} \label{prop:Potential}
Let the energy price rates in (\ref{eqn:EnergyPrice}) satisfy $ \rho^v_{1,k} = \rho_{1,k}>0$ and $\rho_{2,k}^v = \rho_{2,k}\geq0,$ for all $ v\in\ml{A}$. Then, for all  ($\mbf{x}, \phi$) such that $\mc{Z}(\mbf{x},\phi)\neq \varnothing$, the following conditions hold:
\begin{enumerate}[label=(\roman*)]

\item The pseudo-gradient $F(\cdot, \phi)$ defined as in (\ref{eqn:VI}), is strongly monotone~\cite[Def. 2.23]{bauschke2017convex}, i.e., there exists a constant $\mu>0$, such that
$\forall z,z'\in \ml{Z}(\mbf{x},\phi)$ \\$(F(z,\phi)-F(z',\phi))^\top(z-z') \geq \mu \|z-z'\|^2$.

\item There exists a unique v-GNE for the game in (\ref{eqn:MPCagent}).

\item   The game in ($\ref{eqn:MPCagent}$) is a generalized potential game \cite[Def. 2.1]{facchinei2010decomposition} with potential function $P(\cdot, \phi)$ defined as

\begin{align} \label{eqn:Potential}
\mb{P}(z, \phi) = \sum_{k\in \bb{Z}_N} \hat{\ell}_k(x_k,u_k),
\end{align}
where $\hat{\ell}_k(x^v_k,u^v_k) \!=\! \sum_{v\in \ml{A}}\frac{1}{2}\big(
\rho_{1,k}^v \, (l_k^v)^2 
+ \rho_{2,k}^v l_k^v \!+\! \gamma_{1,k}^v \, (\zeta_k^v)^2 \!+\!  \gamma_{2,k}^v \, (q_k^v)^2   \!+\! \ell_k^v(x^v_k,u^v_k,u_k^{-v}) \big)$.
{\hfill $\square$}
\end{enumerate}
\end{prop}
\begin{proof} The proof is given in Appendix \ref{ss:AppendixA}.
\end{proof}
 The assumption that price rates $\rho_1$ and $\rho_2$ are the same means that all houses in a neighbourhood get energy delivered under the same conditions which is a reasonable assumption in the DSM context.

The following theorem gives conditions under which the closed-loop system (\ref{eqn:System}) admits a practically stable \cite[Definition 4.1]{faulwasser2018economic} equilibrium point $(\bar{x},\bar u) = \text{col}(\bar x^v, \bar u^v)_{v\in \ml{A}}$, which is the unique v-GNE of the steady-state game

\begin{align}\label{eqn:EquilSteady}
v\in \mc{A}: \quad
\left\{
\begin{array}{r l}
\displaystyle \min_{x^v\in \ml{X}^v, u^v} & \;\ell^v(x^v,u^v, u^{-v}) \\ 
 \text{s.t} &  x^v =  A^v x^v + B^v u^v + d^v\\ 
& u^v \in \mc{U}^v\cap \ml{C}(\bar{u}^{-v}).
\end{array}  
\right.
\end{align}

\begin{theorem} \label{thrm:Stability}
Suppose that $\rho^v_{1,k} = \rho_{1,k}, \; \rho_{2,k}^v = \rho_{2,k},$ for all $v \in \ml{A}$, and that the exogenous parameters $\phi$ are constant both in time and over the prediction horizon. Next, define $U = \bigtimes_{v\in \mc{A}} \mc{U}^v \cap \mc{C}(u^{-v})$, where we have suppressed all dependencies on $\phi$ to simplify the notation, and the set of all initial conditions that can be driven to $\bar x$, i.e., 

\begin{multline}
    \ml{R} := \big\{\mbf{x}\in \mc{X}~|~ \exists \tilde u = \{u_k\}_{0}^\infty \subseteq U \text{ s.t. } \forall k ~ \geq 0~ \\ x_k(\mbf{x},\tilde u) \in \mc{X} \text{ and } \lim_{k\to\infty}  x_k(\mbf{x},\tilde u) = \bar x \big\},
\end{multline}
where $x_k(\mbf{x},\{u_k\})$ denotes the solution of the collected dynamics $x_{k+1} = Ax_k + Bu_k + d$ of \eqref{eqn:LTIDynamics} with input sequence $\{u_k\}$ and initial condition $\mbf{x}$. $\ml{R}$ is the biggest possible region of attraction given the constraints and it is non-empty if $(\bar{x}, \bar{u}) \in \ml{X} \times U$. Then, there exists $\bar N\in \bb{N}$ such that if $N \geq \bar N$ the strategy profile $\bar x$ is a practically stable equilibrium of the closed-loop system \eqref{eqn:System}. That is, there exist $\eta\in \mc{KL}$ and $\gamma:\bb{R}_{\geq 0} \to \bb{R}_{\geq 0}$ such that for all $\mbf{x} \in \ml{R}$ and $t\geq 0$ the trajectories of the closed-loop system \eqref{eqn:System} satisfy $x_t\in \mc{X}$, $u_t \in U$, and

\begin{equation}
\|x_t - \bar x\| \leq \eta(\|x_t - \bar x\|,t) + \gamma(N),
\end{equation} 
with $\lim_{N\to\infty} \gamma(N) = 0$.
\end{theorem}

\begin{proof} The proof is given in Appendix \ref{ss:AppendixB}.
\end{proof}

This result holds under the assumptions that parameters are time-invariant, i.e., $\phi$ is constant, and that the price rates $\rho_1$ and $\rho_2$ are the same for all prosumers.  It implies that if the closed-loop system is unforced, i.e., external parameters remain constant, the system stabilizes at an equilibrium point. This could happen, e.g, in a calm period without disturbances in which all prices $\rho_1$ and $\rho_2$ and time-varying load and consumption bounds $\bar{L}_t, \underline{L}_t$, etc. remain constant. The equilibrium point is the unique strategically-stable and fair GNE subject to steady-state dynamics and is a desirable operating point of the unforced system. Further, since the potential function (\ref{eqn:Potential}) is strongly monotone and the parameters $\phi$ enter the constraints linearly, we expect a degree of robustness to variation in $\phi$, see \cite{limon2009input}. This is supported by our numerical results in the next section.

\section{Simulation study}
We perform a numerical study in which we demonstrate (i) that our RHG approach outperforms the standard day-ahead optimization in terms of peak load shaving and (ii) that it can enforce the aggregate load constraints despite unforeseen disturbances. All simulations are implemented in Python and v-GNE computations are carried out centrally using \texttt{quadprog} \cite{andersen2022cvxopt}. However, as pointed out in Remark~\ref{remark:information} other information structures are possible.

We assume that all prosumers own the same storage device, i.e., a lithium-ion battery with SoC dynamics as in (\ref{eqn:BatteryDyn}) and parameters $\alpha^v = \sqrt[24]{0.9}$ (which corresponds to a leakage rate of 0.9 over the 24 hours) and $\beta^v = 0.9$ as in \cite{atzeni2013demand}, $\bar{q}^v = 15$ kWh and $\bar{s}^v = -\underline{s}^v = 0.7 \, \bar{q}^v$. The price rates are set as $\rho_1 = 0.015 $ \textdollar/kWh and $\rho_2 = 0.05$ \textdollar/kWh as a base price and the nominal consumption profiles and solar generation profiles were collected between May and November 2019 and are of single-family homes situated in the state of New York \cite{pecan2022residential}. For peak consumption hours, 6:00 to 10:00 and 18:00 to 22:00 the price rates are doubled. The average electricity price for the day-ahead optimization and the RHG scheme is about 0.37 \textdollar/kWh. The energy debt constraint at midnight of every day is set to $-1 \text{kWh}\leq\zeta_N^v\leq1\text{kWh} $.

\subsection{Peak shaving}
In the first case study, we consider 10 active prosumers and 5 passive consumers in the state of New York. In Figure \ref{fig:AggregateLoad}, we compare their aggregate load on the main grid over 48 hours in three different scenarios: no DSM (i.e. the prosumer's load profile equals their nominal load), day-ahead optimization, and RHG. The RHG scheme reduces peak aggregate load by $-49\%$ which is $6\%$ more than the peak shaving achieved by the day-ahead scheme. In fact, unlike the open-loop day-ahead optimization, RHG does not suffer from any ``end-of-day" effects, namely, the fact that prosumers at the end of the day discharge all their batteries and reduce their load on the grid as this is most cost efficient in a finite-horizon scenario.

\begin{figure}

        \includegraphics[width=0.485\textwidth]{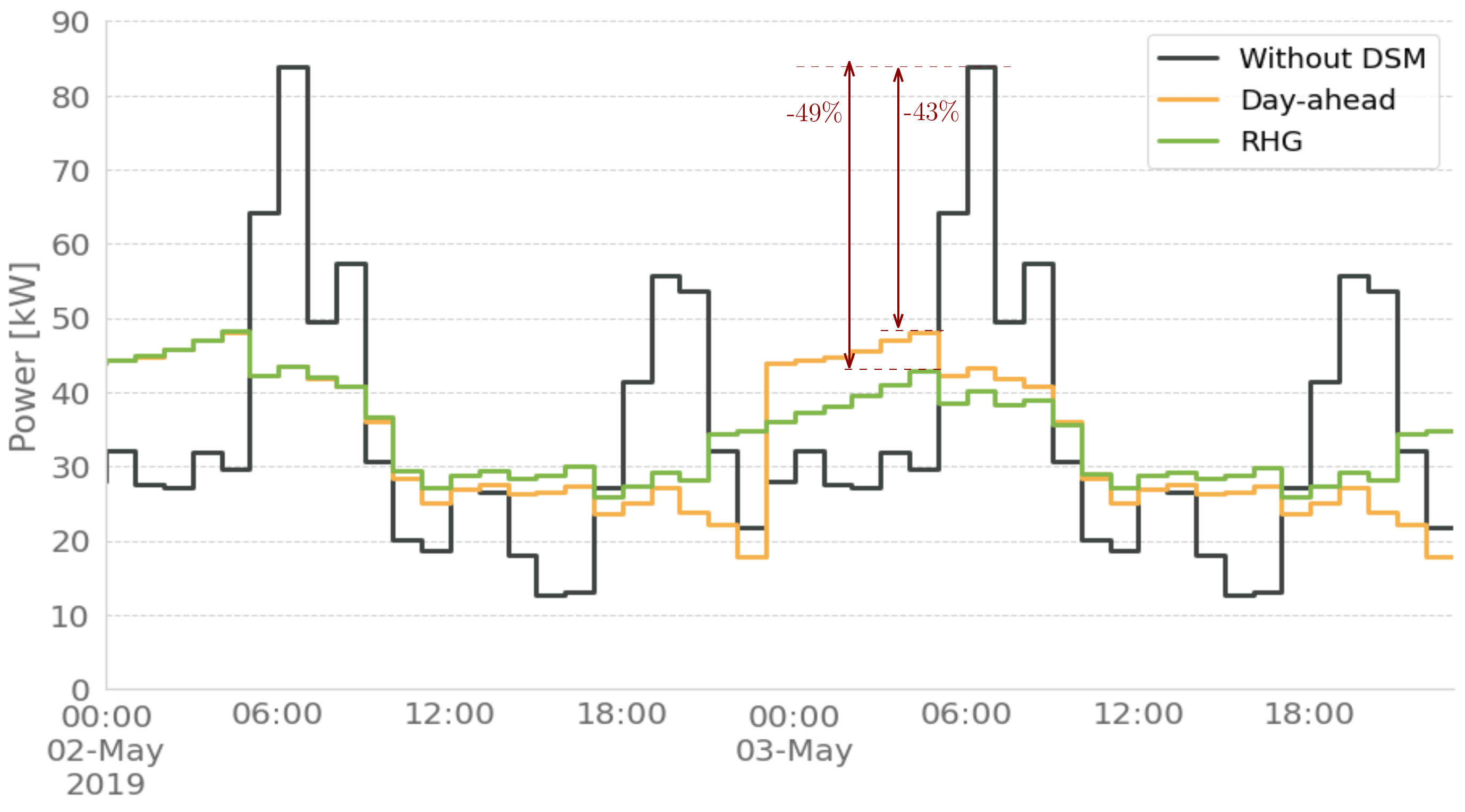} 
        \caption{Aggregate Load comparison for New York over 2 days for a set of 10 active prosumers and 5 passive consumers. Price rates $\rho_1, \;\rho_2$ are doubled from 6:00 to 10:00 and 18:00 to 22:00 to incentivize peak shaving.}
\label{fig:AggregateLoad}
\end{figure} 

\subsection{Disturbance rejection}
A major advantage of closed-loop schemes such as RHG is the immediate response to disturbances and price fluctuations. We model a scenario in which a disturbance, such as a line fault,  leads to a sudden drop in the power available to prosumers from the main grid, i.e., the upper bound $\bar{L}_t$ drops by $-60\%$ for 5 hours from 1:00 to 6:00. The RHG scheme successfully shifts the load of prosumers to later hours in the day as shown in Figure \ref{fig:Disturbance}. This compares favourably with the RHG-DSM scheme in \cite{stephens2015game} which cannot handle coupling constraints and is thus not capable of enforcing aggregate load constraints (with or without disturbances).
\begin{figure}

        \includegraphics[width=0.485\textwidth]{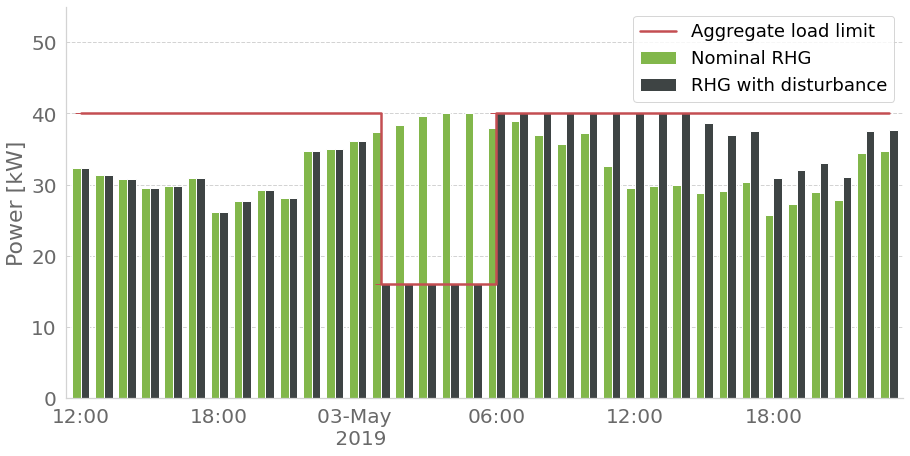} 
        \caption{Disturbance rejection of the RHG for 10 active prosumers and 5 passive consumers. The green bars represent the aggregate load profile of the nominal RHG. The blue bars show that the RHG scheme can successfully react to a sudden -60 $\%$ drop in the aggregate load limit $\bar{L}_t$.}
\label{fig:Disturbance}
\end{figure}

\section{Conclusion}

In this paper, we proposed a closed-loop demand-side management policy based on Receding Horizon Games. The policy coordinates self-interested prosumers who use local storage/generation, and load shifting to collectively enact peak-shaving, enforce system-wide constraints, and reject disturbances. Further, we derived conditions for closed-loop stability and demonstrated the potential of the proposed method through numerical simulations in the case of homogeneous electricity prices for the prosumers. The case of non-symmetric pricing, which is currently not covered by our theory, is a compelling future research direction.

\bibliography{RHG_for_DSM}

\appendix

\subsection{Proof of Proposition \ref{prop:Potential}}\label{ss:AppendixA}

(i): The pseudo-gradient $F(\cdot, \phi)$ defined in (\ref{eqn:VI}) is an affine mapping of the form

\begin{align} \label{eqn:Subdiffz}
\text{F}(z,\phi) = H(\phi) z+ f(\phi)
\end{align}
with   $H_1 = (S^T \rho_1 S) \otimes I+ (S^T \rho_1 S) \otimes \mathds{1}^T \mathds{1}$, $H_2 = 2 \; \text{blkdiag}(Q^1,\dots,Q^M)$ and $H = \text{blkdiag}(H_1, H_2)$ where $S = I_N \otimes   \left[\begin{smallmatrix}
1& 1
\end{smallmatrix}\right] \in \bb{R}^{N\times 2 N} $ is a selection matrix. Further,  $f(\phi) =  \text{col}(S^T\big(  -  (\sum_{s} g^s + g^v)^T \rho_1+\mathds{1}_N^T \rho_2 \big)\mathds{1}_{M}, \mbf{0}_{2M(N+1)})$ is a constant vector. The symmetric matrix $H$ is positive definite as $\rho_1>0$ and $\rho_2\geq 0$, thus $F(\cdot,\phi)$ is strongly monotone~\cite[Def. 2.23]{bauschke2017convex}. 
(ii): Since $F(\cdot, \phi)$ is strongly monotone by (i), we can invoke \cite[Cor. 23.37]{bauschke2017convex} to conclude that $\ml{S}(\mbf{x},\phi)$ in \eqref{eqn:SolMap} is a singleton whenever $\mc{Z}(\mbf{x},\phi)\neq \varnothing$ . (iii): To show that the game is an exact potential, we take the gradient of $\mb{P}(z,\phi)$ in (\ref{eqn:Potential}) and, by performing some algebraic simplifications, we obtain:

\begin{align*}
\nabla_z P(z, \phi) = H(\phi)z + f(\phi),
\end{align*}

\noindent which coincides with the pseudo-gradient in (\ref{eqn:Subdiffz}). Thus, $\mb{P}(z, \phi)$ is an exact potential for the generalized potential game in (\ref{eqn:GamePerAgent}) \cite[Def.~2.1]{facchinei2010decomposition}.{\hfill $\blacksquare$}


\subsection{Proof of Theorem \ref{thrm:Stability}} \label{ss:AppendixB}
The proof is in two steps: 1. We define a MPC problem whose solutions coincide with v-GNEs of the game \eqref{eqn:MPCagent} by exploiting the potential derived in Prop.~\ref{prop:Potential}, then 2. we apply economic MPC stability results to this surrogate problem.

1. We start by defining the surrogate OCP using the stage cost from the potential function in (\ref{eqn:Potential}), i.e.,

\begin{align} \label{eqn:PotentialMPC}
&&\displaystyle V_N(\mbf{x}) = \min_{u,\, x}& \displaystyle \sum_{k\in \bb{Z}_N} \hat{\ell}(x_k, u_k)\\ 
&&\textrm{s.t.} \quad &  x_{k+1} =  A x_k + B u_k + d, \;   k\in \bb{Z}_N \nonumber\\
&&& u_k \in U, \;  k\in \bb{Z}_{N} \nonumber \\
&&& x_k \in \ml{X}_k,~~ x_0 = \mbf{x}, \;   k\in \bb{Z}_{N+1} \nonumber
\end{align}

\noindent If $\rho^v_{1,k} = \rho_{1,k}, \; \rho_{2,k}^v = \rho_{2,k}, \; \forall v \in \ml{A}$, then \eqref{eqn:MPCagent} is an exact potential game by Proposition \ref{prop:Potential} and the unique v-GNE of \eqref{eqn:MPCagent} exactly coincides with the minimizer of (\ref{eqn:PotentialMPC}). Further, the steady state $(\bar{x},\bar u)$ that solves \eqref{eqn:EquilSteady} also satisfies 

\begin{equation}\label{eqn:Blub}
    (\bar{x}, \bar{u})  =  \displaystyle \argmin_{x \,\in\, \ml{X},u \,\in\, U}\; \left \{\hat{\ell}(x,u)~|~ x =  A x + Bu + d\right\}.
\end{equation}

2. We show stability of $\bar{x}$ by exploiting the equivalence between (\ref{eqn:MPCagent}) and (\ref{eqn:PotentialMPC}). Our objective is to apply \cite[Theorem 4.1]{faulwasser2018economic} from economic MPC. To do so, we need to show (i) that there exists a non-negative function $\lambda$ and $\delta \in \mc{K}_\infty$ such that

\begin{equation} \label{eq:dissipativity}
\delta(\|x-\bar x\|)\leq \hat\ell(x,u) - \hat\ell(\bar x,\bar u) + \lambda(x) - \lambda(x^+)
\end{equation}
with $x^+ = Ax + Bu + d$; (ii) the system \eqref{eqn:LTIDynamics} is exponentially reachable \cite[Assumption 4.2]{faulwasser2018economic} on $\ml{R}$; and (iii) that \eqref{eqn:LTIDynamics} is $2M$-step reachable in the sense of \cite[Assumption 4.3]{faulwasser2018economic}.

To show (i), we observe from \eqref{eqn:Potential} that $\hat{\ell}$ is strongly convex, that the dynamics are linear and that the constraint set $\mc{X} \times U$ is polyhedral. Thus, \cite[Prop. 4.3]{damm2014exponential} applies and there exist $a\in \bb{R}^{2M}$ and $b > 0$ such that $\lambda(x) = a^T x$ and $\delta(s) = b s^2$ satisfy \eqref{eq:dissipativity}. Further, since $\mc{X}$ is compact we can make $\lambda$ non-negative by redefining $\lambda(x) =a^T x - \min_{\mc{X}}a^T x$.

Condition (ii) holds if $\forall \mbf{x}\in \ml{R}$ there exists a sequence $\{u_k\}$ such that the solution trajectory $x_k$ satisfies $\|x_k - \bar x\| + \|u_k - \bar u\| \leq c\epsilon^k$ for $c > 0$ and $\epsilon \in (0,1)$. Each pair $(A^v,B^v)$ is controllable and thus the infinite-horizon MPC feedback law \cite{meadows1993receding} is exponentially stabilizing and can be used to construct a suitable sequence $\{u_k\}, \,$ $\forall \mbf{x}\in \ml{R}$. Finally, (iii) follows from linear controllability of each $(A^v,B^v)$. 

Since (i)-(iii) are satisfied, we can invoke \cite[Theorem 4.1]{faulwasser2018economic} to prove practical stability of $\bar x$ and recursive feasibility, if $N$ is chosen large enough. Finally, to show that $\gamma(N) \to 0$ as $N\to \infty$ we exploit \cite[Lemma 4.1]{faulwasser2018economic} for which we need to prove that $\tilde V_N(\mbf{x}) = V_N(\mbf{x}) + \lambda(\mbf{x})$ is uniformly continuous. The OCP \eqref{eqn:PotentialMPC} has a quadratic cost, linear dynamics, and polyhedral constraints, thus $V_N$ is a piecewise quadratic function \cite{tondel2003algorithm}. As $\lambda$ is affine, this implies that $\tilde V_N$ is also piecewise quadratic and, thus, uniformly continuous over the compact domain $\mc{X}$, by the Heine--Cantor theorem. {\hfill $\blacksquare$}

\end{document}